\documentclass[11pt]{article}
\usepackage{geometry}
\usepackage{amsmath, amssymb, amsthm}
\usepackage{graphicx}
\usepackage{bm}
\usepackage[linesnumbered,ruled]{algorithm2e}
\usepackage{fullpage}
\usepackage{mathtools}
\usepackage{url}
\usepackage{color} 

\newtheorem{lemma}{Lemma}

\newtheorem{claim}{Claim}

\title{Polynomial-Time Solvability of 
the Mixed Evacuation Problem}
\author{Naoyuki Kamiyama\thanks{This work was 
supported by JSPS KAKENHI Grant Number JP24K14825.}}
\date{Graduate School of Informatics\\ 
Kyoto University \\ 
Kyoto, Japan \\
\url{kamiyama@amp.i.kyoto-u.ac.jp}}

\begin{document}

\maketitle

\begin{abstract}
The mixed evacuation problem is a variant of 
the dynamic transshipment problem where there exist 
two modes of transportation and we allocate 
supplies to each mode of transportation in such a way that 
we can send all the supplies 
from sources to sinks within a given time limit.
It is known that this problem can be solved in polynomial time 
when the number of sources and sinks is small. 
In this paper, we prove that 
the mixed evacuation problem can be solved in 
polynomial time with the ellipsoid method. 
\end{abstract}

\section{Introduction} 

A dynamic network flow, which was introduced by Ford and
Fulkerson~\cite{FordF58,FordF62}, 
is a flow defined in
a dynamic network, i.e., a directed
graph with capacities and transit times on its arcs. 
The dynamic transshipment problem is a problem in 
a dynamic network with multiple sources and multiple sinks, and 
the goal of the dynamic transshipment problem
is to determine whether we can send all the supplies 
from sources to sinks within a given time limit. 
It is known that the dynamic transshipment problem
can be solved in polynomial time~\cite{HoppeT00}. 

The mixed evacuation problem~\cite{HanawaHKKT18} is a variant of 
the dynamic transshipment problem where there exist 
two modes of transportation and we allocate 
supplies to each mode of transportation
in such a way that 
we can send all the supplies 
from sources to sinks within a given time limit. 
It is known that this problem can be solved in polynomial time 
when the number of sources and sinks is small~\cite{HanawaHKKT18}.
In this paper, we prove that 
the mixed evacuation problem can be solved in 
polynomial time with the ellipsoid method. 

\paragraph{AI disclosure.}
The main idea of the proof was obtained by using GPT-5.6 Luna and GPT-6 Astra. 
More precisely, the idea that the dual problem \eqref{eq:dual} is used to prove 
Lemmas~\ref{lemma_1} and \ref{lemma_2} was obtained 
by using ChatGPT.
The author independently verified 
the mathematical correctness and 
reorganized the proof.
In addition, ChatGPT was used for proofreading. 
The author 
takes full responsibility for the 
correctness and the content of this paper.

\section{Preliminaries}
\label{section:preliminaries} 

Let $\mathbb{R}$, $\mathbb{R}_+$, 
$\mathbb{Z}_+$, and $\mathbb{Z}_-$ denote  
the sets of reals, 
non-negative reals, non-negative integers, 
and non-positive integers, respectively.  
For each finite set $N$, each vector $x \in \mathbb{R}^N$, and 
each subset $M \subseteq N$, 
we define 
$x(M) := \sum_{e \in M}x(e)$.
Furthermore, for each finite set $N$ and each 
pair of vectors $x,y \in \mathbb{R}^N$, 
we define $\langle x,y \rangle := \sum_{e \in N}x(e)y(e)$. 

The \emph{mixed evacuation problem} is defined as follows. 
In this problem, we are given a finite directed graph 
$D=(V,A)$ and two disjoint subsets $S,T \subseteq V$. 
The subset $S$ (resp.\ $T$) is the set of source vertices 
(resp.\ sink vertices) in $V$.   
We assume that no arc in $A$ enters (resp.\ leaves) 
a vertex in $S$ (resp.\ $T$).
In addition, in this problem, 
we are given 
an arc capacity vector $c \in \mathbb{Z}_+^A$, 
a supply vector $b \in \mathbb{Z}_+^S$, 
a sink capacity vector $u \in \mathbb{Z}_-^T$, 
transit time vectors  
$\tau_1, \tau_2 \in \mathbb{Z}_+^A$, and 
coefficients $q_1,q_2 \in \mathbb{Z}_+$.
Finally, we are given a time limit ${\rm \Theta} \in \mathbb{Z}_+$.
Define $[{\rm \Theta}] := \{0,1,\ldots, {\rm \Theta}\}$. 
All numerical input data are encoded in binary. 
Define $L$ as the encoding length of the instance of the mixed evacuation problem. 

For each integer $i \in \{1,2\}$, 
each function $f \colon A \times [{\rm \Theta}] \to \mathbb{R}_+$, 
each vertex $v \in V$, and each integer $\theta \in [{\rm \Theta}]$, 
we define 
\begin{equation*}
\partial_if(v, \theta) := 
\sum_{a \in \delta_A(v)}\sum_{t = 0}^{\theta} f(a,t)
-
\sum_{a \in \varrho_A(v)}\sum_{t = 0}^{\theta - \tau_i(a)} f(a,t),
\end{equation*}
where $\delta_A(v)$ (resp.\ $\varrho_A(v)$) 
represents the set of 
arcs in $A$ leaving (resp.\ entering) $v$.
In other words, $\partial_if(v, \theta)$ represents the difference between 
the total amount of flow leaving $v$ by the time step $\theta$ and 
the total amount of flow entering $v$ by the time step $\theta$. 

A vector $d \in \mathbb{R}^{S \cup T}$ is called 
an {\it allocation} if 
$d(v) \ge 0$ for every vertex $v \in S$ and 
$d(v) \le 0$ for every vertex $v \in T$.
For each integer $i \in \{1,2\}$, 
each allocation $d \in \mathbb{R}^{S \cup T}$, 
and each vector $w \in \mathbb{R}^A_+$, 
$d$ is said to be {\it $(i,w)$-feasible} if 
there exists a function $f \colon A \times  [{\rm \Theta}] \to \mathbb{R}_+$ 
satisfying the following conditions. 
\begin{description}
\setlength{\parskip}{0mm} 
\setlength{\itemsep}{1.2mm} 
\item[(D1)]
Let $a$ and $\theta$ be an arc in $A$ and an integer in $[{\rm \Theta}]$, respectively. 
\begin{itemize}
\setlength{\parskip}{0mm} 
\setlength{\itemsep}{1.2mm} 
\item
If $\theta \le {\rm \Theta} - \tau_i(a)$, then 
$f(a, \theta) \le q_i \cdot w(a)$.
\item
If $\theta > {\rm \Theta} - \tau_i(a)$, then 
$f(a, \theta) = 0$.
\end{itemize}
\item[(D2)]
Let $v$ and $\theta$ be a vertex in $V$ and an 
integer in $[{\rm \Theta}]$, respectively. 
\begin{itemize}
\setlength{\parskip}{0mm} 
\setlength{\itemsep}{1.2mm} 
\item
If $v \in V \setminus (S \cup T)$, then 
$\partial_if(v,{\rm \theta}) \le 0$. 
\item
If $v \in S$ (resp.\ $T$), then 
$\partial_if(v,{\rm \theta}) \le d(v)$
(resp.\ $\partial_if(v,{\rm \theta}) \ge d(v)$). 
\end{itemize}
\item[(D3)]
Let $v$ be a vertex in $V$. 
\begin{itemize}
\setlength{\parskip}{0mm} 
\setlength{\itemsep}{1.2mm} 
\item
If $v \in V \setminus (S \cup T)$, then 
$\partial_if(v,{\rm \Theta}) = 0$. 
\item
If $v \in S$, then 
$\partial_if(v,{\rm {\rm \Theta}}) = d(v)$.
\end{itemize}
\end{description}

For each integer $i \in \{1,2\}$ and 
each vector $w \in \mathbb{R}^A_+$, 
let $\mathcal{F}_i(w)$ be the set of 
$(i,w)$-feasible allocations in $\mathbb{R}^{S \cup T}$. 
An {\it assignment} is a 
tuple $(d_1,d_2,w_1,w_2)$
such that 
$d_1,d_2$ are allocations in $\mathbb{R}^{S \cup T}$ and 
$w_1,w_2 \in \mathbb{R}^A_+$.
Furthermore, an assignment $(d_1,d_2,w_1,w_2)$
is said to be {\it feasible} if 
it satisfies the following conditions.
\begin{description}
\setlength{\parskip}{0mm} 
\setlength{\itemsep}{1.2mm} 
\item[(F1)]
For every vertex $v \in S$ (resp.\ $T$),
$d_1(v) + d_2(v) = b(v)$ (resp.\ $\ge u(v)$).
\item[(F2)]
For every arc $a \in A$, 
we have 
$w_1(a) + w_2(a) \le c(a)$.
\item[(F3)]
$d_1 \in \mathcal{F}_1(w_1)$ 
and 
$d_2 \in \mathcal{F}_2(w_2)$. 
\end{description}

The goal of the mixed evacuation problem 
is to decide whether there exists a feasible assignment.
It is 
not difficult to see that 
we can straightforwardly formulate 
the mixed evacuation problem
as a linear program.
However, the size of 
this linear program depends on $\Theta$.
Since $\Theta$ is encoded in binary, 
the size of 
this linear program may not be bounded 
by a 
polynomial in $L$.

\section{Useful Facts} 
\label{section:useful}

In this section, we introduce 
facts that are needed in Section~\ref{section:polytime}.

Assume that we are given an integer $i \in \{1,2\}$, 
a vector $w \in \mathbb{R}_+^A$, and 
a subset $X \subseteq S \cup T$.
We define ${\bf D}_{i,w}(X)$ as the set of functions 
$f \colon A\times [{\rm \Theta}] \to \mathbb{R}_+$
satisfying (D1) and the following conditions (D4) and (D5). 
\begin{description}
\setlength{\parskip}{0mm} 
\setlength{\itemsep}{1.2mm} 
\item[(D4)]  
Let $v$ and $\theta$ be a vertex in $V$ and an 
integer in $[{\rm \Theta}]$, respectively. 
\begin{itemize}
\setlength{\parskip}{0mm} 
\setlength{\itemsep}{1.2mm} 
\item
If $v \in V \setminus X$, then 
$\partial_if(v,\theta) \le 0$. 
\item
If $v \in X$, then 
$\partial_if(v,\theta) \ge 0$. 
\end{itemize}
\item[(D5)]
For every vertex $v \in V \setminus (S \cup T)$, 
we have 
$\partial_if(v, {\rm \Theta}) = 0$.
\end{description}
Recall that 
no arc in $A$ enters (resp., leaves) 
a vertex in $S$ (resp., $T$).
Thus, for every function $f \in {\bf D}_{i,w}(X)$ and
every vertex $v \in (S \setminus X) \cup (T \cap X)$, 
$\partial_if(v, {\rm \Theta}) = 0$ holds.  
Intuitively speaking, an element in ${\bf D}_{i,w}(X)$ is a 
dynamic flow from $S \cap X$ to 
$T \setminus X$ without supplies and demands. 
Furthermore, we define the function 
$o_{i,w} \colon 2^{S \cup T} \to \mathbb{R}_+$ by   
\begin{equation*}
o_{i,w}(X) := \max\Big\{\sum_{v \in X}\partial_if(v,{\rm \Theta}) 
\ \Big| \ f \in {\bf D}_{i,w}(X)\Big\}.
\end{equation*}
Intuitively speaking, $o_{i,w}(X)$ is the maximum flow value 
of a dynamic flow from $S \cap X$ to 
$T \setminus X$. 
The following result was cited by Hoppe and Tardos~\cite{HoppeT00}
as personal communication.

\begin{lemma}[Klinz~{\cite[Theorem~5.1]{HoppeT00}}] \label{theorem:K00}
Assume that we are given an integer $i \in \{1,2\}$,  
an allocation $d \in \mathbb{R}^{S\cup T}$,
and a vector $w \in \mathbb{R}_+^A$.
Then
$d \in \mathcal{F}_i(w)$
if and only if 
$d(X) \le o_{i,w}(X)$
for every subset $X \subseteq S \cup T$.
\end{lemma}

In what follows, let $r$ be a new vertex that is not contained in $V$. 
Define $U \coloneqq V \cup \{r\}$. 
For each vertex $v \in S$ (resp.\ $T$), 
we define $a_{r,v}$ (resp.\ $a_{v,r}$) as an
arc from $r$ to $v$
(resp.\ $v$ to $r$). 

For each subset $X \subseteq S \cup T$, 
we introduce the following definitions. 
Define 
\begin{equation*}
E_X := A \cup \{a_{r,v} \mid v \in S \cap X\} \cup \{a_{v,r} \mid v \in T \setminus X\}.
\end{equation*}
Define the directed graph $H_X$ as 
$(U,E_X)$.
Then for each integer $i \in \{1,2\}$, each vector 
$w \in \mathbb{R}_+^A$, 
and each vector $\xi \in \mathbb{R}^{E_X}_+$,
$\xi$ is called a 
{\it feasible static flow} with respect to $(i,w,X)$ if 
it satisfies the following conditions (S1) and (S2).
\begin{description}
\setlength{\parskip}{0mm} 
\setlength{\itemsep}{1.2mm} 
\item[(S1)]
For every arc $a \in A$, we have 
$\xi(a) \le q_i \cdot w(a)$.
\item[(S2)]
For every vertex $v \in U$, 
we have 
$\xi(\delta_{E_X}(v)) = \xi(\varrho_{E_X}(v))$.
\end{description} 

For each integer $i \in \{1,2\}$ and 
each vector $w \in \mathbb{R}_+^A$, 
we define ${\bf S}_{i,w}(X)$ as 
the set of feasible static flows with 
respect to $(i,w,X)$. 
For each integer $i \in \{1,2\}$, 
we define 
the vector $k_i^X \in \mathbb{R}^{E_X}$ as follows. 
For each arc $a \in A$,
we define $k_i^X(a):= -\tau_i(a)$. 
For each arc $a \in E_X \setminus A$, 
if $a = a_{r,v}$ for some vertex $v \in S$, 
then we define 
$k_i^X(a):= 0$. 
Furthermore, for each arc $a \in E_X \setminus A$, 
if $a = a_{v,r}$ for some 
vertex $v \in T$, 
then we define 
$k_i^X(a):= {\rm \Theta} + 1$. 

\begin{lemma}[Ford and Fulkerson~\cite{FordF58,FordF62}] \label{theorem:FF58}
For every integer $i \in \{1,2\}$, every vector $w \in \mathbb{R}_+^A$, 
and every subset $X \subseteq S \cup T$, 
the value $o_{i,w}(X)$ is equal to 
the optimal objective value of the problem of 
maximizing 
$\langle k_i^X, \xi \rangle$ under the condition that 
$\xi \in {\bf S}_{i,w}(X)$. 
\end{lemma}

\section{Polynomial-Time Solvability}
\label{section:polytime}

In this section, we prove that 
the mixed evacuation problem can be solved in polynomial time. 

Define $\mathbb{V} \coloneqq (\mathbb{R}^{S \cup T})^2 \times 
(\mathbb{R}^A)^2$. 
In addition, we define 
${\bf K}$ as the 
set of tuples $(d_1,d_2,w_1,w_2) \in \mathbb{V}$ 
satisfying the 
following constraints. 
\begin{equation} \label{eq:K}
\begin{split}
& d_i(v) \ge 0 
\ \ \ \ (i \in \{1,2\}, v \in S), \ \ \  
 d_i(v) \le 0 
\ \ \ \ (i \in \{1,2\}, v \in T) \\
& w_i(a) \ge 0 
\ \ \ \ (i \in \{1,2\}, a \in A) \\
& d_1(v) + d_2(v) = b(v)  
\ \ \ \ (v \in S), \ \ \  
 d_1(v) + d_2(v) \ge u(v)  
\ \ \ \ (v \in T) \\
& w_1(a) + w_2(a) \le c(a)  
\ \ \ \ (a \in A) \\
& d_i(X) \le o_{i,w_i}(X) 
\ \ \ \ (i \in \{1,2\}, X \subseteq S \cup T).
\end{split}
\end{equation}
The mixed evacuation problem is equivalent 
to the problem of determining 
whether ${\bf K} \neq \emptyset$. 

It follows from 
\cite[Theorem~6.4.9]{GrotschelLS93} that, in order to 
prove that the mixed evacuation problem can be solved in polynomial time, 
it is sufficient to prove the following two lemmas. 

\begin{lemma} \label{lemma_1}
There exists a finite system of
inequalities with rational coefficients such that its solution set is equal to ${\bf K}$ 
and the encoding length of each inequality of the system is bounded by a polynomial in $L$. 
\end{lemma} 

\begin{lemma} \label{lemma_2}
Assume that we are given a rational element 
$(\widehat{d}_1,\widehat{d}_2,\widehat{w}_1,\widehat{w}_2) \in \mathbb{V}$.
Then the separation problem for ${\bf K}$ and 
$(\widehat{d}_1,\widehat{d}_2,\widehat{w}_1,\widehat{w}_2)$
can be solved in time bounded by a polynomial in $L$ and $B$, where 
$B$ denotes the encoding length of $(\widehat{d}_1,\widehat{d}_2,\widehat{w}_1,\widehat{w}_2)$. 
\end{lemma} 

The linear program in Lemma~\ref{theorem:FF58}
can be described as follows. 
\begin{equation} \label{eq:primal}
\begin{array}{cl}
\mbox{Maximize} & 
\ 
\langle k_i^X, \xi \rangle 
\vspace{1mm} \\
\mbox{subject to} & 
\ 
\xi(a) \le q_i \cdot w(a) \ \ \ (a \in A)
\vspace{1mm} \\
& 
\ 
\xi(\delta_{E_X}(v)) = \xi(\varrho_{E_X}(v))
\ \ \ (v \in U)
\vspace{1mm} \\
& 
\ 
\xi \in \mathbb{R}_+^{E_X}. 
\end{array}
\end{equation}
Thus, the dual problem of \eqref{eq:primal}
can be described as follows. 
\begin{equation} \label{eq:dual}
\begin{array}{cl}
\mbox{Minimize} & 
\ 
q_i \cdot \langle w, \alpha \rangle
\vspace{1mm} \\
\mbox{subject to} & 
\ 
\alpha(a) 
+ \beta(\partial^-a)
- \beta(\partial^+a) 
\ge k_i^X(a)
\ \ \ (a \in A)
\vspace{1mm} \\
& 
\ 
\beta(\partial^-a)
- \beta(\partial^+a) 
\ge k_i^X(a)
\ \ \ (a \in E_X \setminus A)
\vspace{1mm} \\
& 
\ 
(\alpha, \beta) \in \mathbb{R}_+^A 
\times \mathbb{R}^U, 
\end{array}
\end{equation}
where
$\partial^+ a$ 
(resp.\ $\partial^- a$) 
represents the tail 
(resp.\ head) of $a$
for each arc $a \in E_X$. 
Notice that since the feasible region of \eqref{eq:primal}
is bounded and has a trivial solution, 
\eqref{eq:primal} has an optimal solution.
Thus, 
the duality theorem of linear programming (see, e.g., \cite[Theorem~3.7]{ConfortiCZ14}) 
implies that 
\eqref{eq:dual} has an optimal solution. 
It should be noted that the constraints of \eqref{eq:dual} 
do not depend on $w$. 

Define $\mathcal{C}_X$ as the set of weakly connected components of $H_X$. 
That is, $\mathcal{C}_X$ is the family of maximal subsets $C \subseteq U$ 
such that $C$ is connected in the undirected graph obtained by ignoring the 
orientations of arcs in $E_X$. 
Then for each element $C \in \mathcal{C}_X$, we select an arbitrary vertex 
in $C$ as a special vertex $v_C$. 

It is not difficult to see that even
if we add the constraint that 
$\beta(v_C) = 0$ for each element $C \in \mathcal{C}_X$ 
to \eqref{eq:dual}, 
the optimal objective value does not change. 
(If $\beta(v_C) \neq 0$, then we subtract $\beta(v_C)$ from 
$\beta(v)$ for each vertex $v \in C$.)
Define ${\bf T}_i(X)$ as 
the set of feasible solutions to \eqref{eq:dual} 
with the constraint that 
$\beta(v_C) = 0$ for each element $C \in \mathcal{C}_X$.
(Recall that the set of feasible solutions of \eqref{eq:dual} 
does not depend on $w$.)
Notice that ${\bf T}_i(X) \neq \emptyset$. 

\begin{lemma} \label{lemma_3}
${\bf T}_i(X)$ is pointed, i.e., 
${\bf T}_i(X)$ does not contain a line. 
\end{lemma}
\begin{proof}
Let $(\widehat{\alpha},\widehat{\beta})$ be an 
element in ${\bf T}_i(X)$, and let 
$(d_1,d_2)$ be a vector in $\mathbb{R}^A \times \mathbb{R}^U$
such that some component of $(d_1,d_2)$ is not zero.
Assume that 
\begin{equation} \label{eq_1:line} 
(\widehat{\alpha},\widehat{\beta}) + \lambda \cdot (d_1,d_2) \in 
{\bf T}_i(X)
\end{equation}
for every real $\lambda \in \mathbb{R}$. 
Then we derive a contradiction. 
Define 
$(\alpha_{\lambda},\beta_{\lambda}) \coloneqq
(\widehat{\alpha},\widehat{\beta}) + \lambda \cdot (d_1,d_2)$. 

Assume that there exists an arc $a \in A$ such that 
$d_1(a) \neq 0$. 
Define 
\begin{equation*}
\lambda \coloneqq \frac{- \widehat{\alpha}(a) - 1}{d_1(a)}.
\end{equation*}
Then we have 
\begin{equation*}
\alpha_{\lambda}(a) = 
\widehat{\alpha}(a) + \lambda \cdot d_1(a) 
= -1 < 0. 
\end{equation*}
This contradicts \eqref{eq_1:line}.
Thus, we can assume that 
$d_1(a) = 0$ for every arc $a \in A$. 

\begin{claim} \label{claim_1:lemma_3}
For every arc $a \in E_X$, 
$d_2(\partial^-a) = d_2(\partial^+a)$. 
\end{claim}
\begin{proof}
Assume that there exists an arc $a \in E_X$ such that 
$d_2(\partial^-a) \neq d_2(\partial^+a)$. 
If $a \in A$, then we define $\Delta \coloneqq \widehat{\alpha}(a)$.
Otherwise, we define $\Delta \coloneqq 0$. 
Define 
\begin{equation*}
\lambda \coloneqq 
\frac{k_i^X(a) - \Delta  
- \widehat{\beta}(\partial^-a)
+ \widehat{\beta}(\partial^+a) - 1}{d_2(\partial^-a) - d_2(\partial^+a)}. 
\end{equation*}
If $a \in A$, then we define $\Delta_{\lambda} \coloneqq \alpha_{\lambda}(a)$.
Otherwise, we define $\Delta_{\lambda} \coloneqq 0$. 
Then we have 
\begin{equation*}
\begin{split}
& \Delta_{\lambda} 
+ \beta_{\lambda}(\partial^-a)
- \beta_{\lambda}(\partial^+a) \\
& = \Delta  
+ \widehat{\beta}(\partial^-a)
- \widehat{\beta}(\partial^+a) 
+ \lambda \cdot (d_2(\partial^- a) - d_2(\partial^+a)) \\
& = \Delta   
+ \widehat{\beta}(\partial^-a)
- \widehat{\beta}(\partial^+a) 
+ k_i^X(a) - \Delta  
- \widehat{\beta}(\partial^-a)
+ \widehat{\beta}(\partial^+a) - 1
< k_i^X(a).
\end{split} 
\end{equation*}
This contradicts \eqref{eq_1:line}.
\end{proof} 

Recall that some component of $(d_1,d_2)$ is not zero.
Since $d_1(a) = 0$ for every arc $a \in A$, 
there exists a vertex $v \in U$ such that 
$d_2(v) \neq 0$. 
Let $C$ be the element in $\mathcal{C}_X$ containing $v$.
Then Claim~\ref{claim_1:lemma_3} implies that 
$d_2(v_C) = d_2(v) \neq 0$. 
Define $\lambda \coloneqq 1$.
Then since $\widehat{\beta}(v_C) = 0$, 
\begin{equation*}
\beta_{\lambda}(v_C) = \widehat{\beta}(v_C) + d_2(v_C) \neq 0.
\end{equation*} 
This contradicts \eqref{eq_1:line}.
\end{proof}

Recall that ${\bf T}_i(X) \neq \emptyset$.
Thus, 
Lemma~\ref{lemma_3} implies that 
${\bf T}_i(X)$ has an extreme point
(see, e.g., \cite[Theorem~3.33 and p.109]{ConfortiCZ14}).
Define ${\rm ex}({\bf T}_i(X))$ as 
the set of extreme points of ${\bf T}_i(X)$. 
Notice that 
${\rm ex}({\bf T}_i(X))$ is finite
(see, e.g., \cite[Theorem~3.34]{ConfortiCZ14}). 

\subsection{Proof of Lemma~\ref{lemma_1}}
\label{section:proof_lemma_1}

Let $i$ be an integer in $\{1,2\}$. 
Then we consider the following inequalities in \eqref{eq:K}. 
\begin{equation} \label{eq_1:eq_K}
\forall X \subseteq S\cup T \colon
d_i(X) \le o_{i,w_i}(X).
\end{equation}
Lemma~\ref{theorem:FF58} implies that 
\eqref{eq_1:eq_K}
is equivalent to 
\begin{equation} \label{eq_2:eq_K}
\forall X \subseteq S\cup T \colon
d_i(X) \le \max\{\langle k_i^X, \xi \rangle \mid \xi \in {\bf S}_{i,w_i}(X)\}.
\end{equation}
By using \eqref{eq:dual} and 
the duality theorem of linear programming (see, e.g., \cite[Theorem~3.7]{ConfortiCZ14}), 
\eqref{eq_2:eq_K} is equivalent to 
\begin{equation} \label{eq_3:eq_K}
\forall X \subseteq S\cup T \colon
d_i(X) \le \min\{q_i \cdot  \langle w_i, \alpha \rangle
\mid (\alpha, \beta) \in {\bf T}_i(X)\}.
\end{equation}
It is known that 
there exists an extreme-point optimal solution to a linear program when 
this linear program has an optimal solution and 
its feasible region is pointed 
(see, e.g., \cite[p.15]{GrotschelLS93}).
Thus, \eqref{eq_3:eq_K}
is equivalent to 
\begin{equation} \label{eq_4:eq_K}
\forall X \subseteq S\cup T \colon
d_i(X) \le \min\{q_i \cdot  \langle w_i, \alpha \rangle
\mid (\alpha, \beta) \in {\rm ex}({\bf T}_i(X))\}. 
\end{equation}
Clearly, \eqref{eq_4:eq_K}
is equivalent to 
\begin{equation*}
\forall X \subseteq S\cup T,
\forall (\alpha, \beta) \in {\rm ex}({\bf T}_i(X)) \colon
d_i(X) \le q_i \cdot  \langle \alpha, w_i \rangle.
\end{equation*}
The encoding lengths of the coefficients and the right-hand sides of the constraints of 
\eqref{eq:dual} 
are bounded by a polynomial in $L$.
Thus, for every element 
$(\alpha, \beta) \in {\rm ex}({\bf T}_i(X))$,
the encoding length of $\alpha$ is bounded by a polynomial in $L$~\cite[Theorem~3.38]{ConfortiCZ14}. 
This completes the proof. 

\subsection{Proof of Lemma~\ref{lemma_2}}

For the constraints of ${\bf K}$ except the last family of constraints, 
we can determine whether these constraints are satisfied 
in time bounded by a polynomial in $L$ and $B$. 
If some constraint among these constraints is not satisfied, then 
we can straightforwardly obtain a desired separating hyperplane.
For every integer $i \in \{1,2\}$, 
since $o_{i,\widehat{w}_i}(\cdot)$ is submodular~\cite{HoppeT00}, 
we can determine whether 
$\widehat{d}_i(X) \le o_{i,\widehat{w}_i}(X)$ 
for every subset $X \subseteq S \cup T$ by using 
a submodular function minimization algorithm~\cite{IwataFF01,Schrijver00}
in time bounded by a polynomial in $L$ and $B$.
Notice that \eqref{eq:primal} implies that 
we can compute $o_{i,\widehat{w}^i}(X)$ for each subset $X \subseteq S \cup T$
in time bounded by a polynomial in $L$ and $B$. 

In what follows, we assume that we have 
an integer $i \in \{1,2\}$ and a subset $Z \subseteq S \cup T$ such that 
$\widehat{d}_i(Z) > o_{i,\widehat{w}^i}(Z)$.
In order to obtain a desired separating hyperplane, 
we first compute an extreme-point optimal solution $(\widehat{\alpha}, \widehat{\beta})$ 
to the following problem.
\begin{equation} \label{eq:lp_T}
\mbox{Minimize} \ \ \ 
q_i \cdot \langle \widehat{w}_i, \alpha \rangle
\ \ \ \mbox{subject to} 
\ \ \ (\alpha, \beta) \in {\bf T}_i(Z).
\end{equation}
Since the encoding length of \eqref{eq:lp_T} is bounded by a polynomial 
in $L$ and $B$, we can find such an optimal solution  
in time bounded by a polynomial in $L$ and $B$~\cite{Khachiyan80}. 
The encoding length of $\widehat{\alpha}$ is 
bounded by a polynomial in $L$ (see the last part of Section~\ref{section:proof_lemma_1}). 
Then we have 
\begin{equation*}
\widehat{d}_i(Z) > 
o_{i,\widehat{w}_i}(Z)
= q_i \cdot  \langle \widehat{\alpha},\widehat{w}_i \rangle.
\end{equation*}

Let $(d^{\prime}_1,d^{\prime}_2,w^{\prime}_1,w^{\prime}_2)$ be 
an arbitrary element in ${\bf K}$. 
Since $w^{\prime}_i$ does not affect the constraints of 
\eqref{eq:dual}, 
$(\widehat{\alpha}, \widehat{\beta})$ is also a feasible solution to 
\eqref{eq:dual}
with respect to $w^{\prime}_i,Z$.
Thus, 
\begin{equation*}
o_{i,w^{\prime}_i}(Z)
\le q_i \cdot \langle \widehat{\alpha}, w^{\prime}_i \rangle.
\end{equation*}
This implies that since
$(d^{\prime}_1,d^{\prime}_2,w^{\prime}_1,w^{\prime}_2) \in {\bf K}$, we have 
\begin{equation*}
d^{\prime}_i(Z) \le o_{i,w^{\prime}_i}(Z) \le 
q_i \cdot \langle \widehat{\alpha}, w^{\prime}_i \rangle. 
\end{equation*}
Thus, 
$d_i(Z) = q_i \cdot \langle \widehat{\alpha}, w_i \rangle$
is a desired separating hyperplane. 
This completes the proof. 

\bibliographystyle{plain}
\bibliography{mixed_evacuation_bib}

\end{document}